\pdfoutput=1

\documentclass[11pt]{article}

\usepackage{amssymb,amsmath,amsthm}	
\usepackage{a4wide}
\usepackage{graphicx}      										
\usepackage{tikz}          										
\usepackage[numbers]{natbib}
\usepackage{dsfont}
\usepackage{amsmath}
\usepackage{xspace}
\usepackage[T1]{fontenc}
\usepackage[utf8x]{inputenc}
\usepackage{graphicx}
\usepackage{caption}
\usepackage{subcaption}

\newcommand{\p}[2]{\ensuremath{\frac{\partial #1}{\partial #2 }}}

\newcommand{\absatz}{\\[12pt]}           
\newcommand{\klabsatz}{\\[8pt]}          
\newtheorem{definition}{Definition}[]
\newtheorem{lemma}{Lemma}[]
\newtheorem{theorem}{Theorem}[]
\newtheorem{remark}{Remark}[]

\newcommand{\R}{\ensuremath{\mathbb{R}}}

\newcommand{\N}{\ensuremath{\mathbb{N}}}
\DeclareMathOperator{\Unif}{Unif}

\begin{document}
	
	\title{Mean Field Limit of a Behavioral Financial Market Model}
	
	\author{Torsten Trimborn\footnote{MathCCES, RWTH Aachen, 52056 Aachen, Germany} \footnote{Corresponding author: trimborn@mathcces.rwth-aachen.de}, Martin Frank\footnote{Karlsruhe Institute of Technology, Steinbuch Center for Computing, Hermann-von-Helmholtz-Platz 1, 76344 Eggenstein-Leopoldshafen, Germany}, Stephan Martin\footnote{Im Hainzenthal 27, 67722 Winnweiler, Germany} }

	\maketitle
	
\begin{abstract}
In the past decade there has been a growing interest in agent-based econophysical financial market models.
The goal of these models is to gain further insights into stylized facts of financial data. We derive the mean field limit of the econophysical Cross model \cite{cross2005threshold} and show that the kinetic limit is a good approximation of the original model. Our kinetic model is able to replicate some of the most prominent stylized facts, namely fat-tails of asset returns, uncorrelated stock price returns and volatility clustering. Interestingly, psychological misperceptions of investors can be accounted to be the origin of the appearance of stylized facts. The mesoscopic model allows us to study the model analytically. We derive steady state solutions and entropy bounds of the deterministic skeleton. These first analytical results already guide us to explanations for the complex dynamics of the model.  
\end{abstract}

\section{Introduction}
In the past years, there has been a number of financial crises (Black Monday 1987, Dot-com Bubble 2000, Global Financial Crisis 2007). Unfortunately, these crises all have in common that classical financial market models fail to explain their origin and existence \cite{colander2009financial, farmer2009economy}. Additionally, these models fail to explain the existence of \textit{stylized facts}, which are assumed to be one important aspect to the creation of market crashes \cite{lebaron2006agentstyl}. Stylized facts are statistical properties of financial data observable all over the world \cite{cont2001empirical}. The most prominent examples are \textit{fat-tails} in asset returns and \textit{volatility clustering} \cite{cont2007volatility, bouchaud2001power}. In physics, stylized facts might be viewed as scaling laws \cite{lux2008stochastic}, which is the reason why physicists became more and more interested in economic models \cite{stanley1999econophysics}.\klabsatz
 This has lead to the new field of research called econophysics which can be traced back to the Dow Jones crash (Black Monday) in 1987. Generally speaking, physicists and economists apply physical theories such as kinetic theory, mean field theory or percolation theory to economic issues. One tool of econophysics are so called agent-based financial market models. Many researchers believe that these models help to gain more insights into financial markets \cite{farmer2009economy, lux2008stochastic}. These modern models of financial markets consist of many interacting agents which are studied with the help of Monte Carlo simulations \cite{lebaron2006agent}.\klabsatz
These models do not consider rational financial agents, often called \textit{homo oeconomicus}, which have been considered in the classical financial market models. They rather consider so called bounded rational agents in the sense of Simon \cite{simon1972theories} and are often inspired by the prospect theory founded by Kahnemann and Tversky \cite{kahneman1979prospect}. These modern financial market models can reproduce stylized facts and they seem to indicate that psychological misperceptions of agents are one reason for their appearance. However, until now the origin of stylized facts is not completely understood \cite{lux2008stochastic}.\klabsatz
Time continuous, in particular kinetic partial differential equations (PDEs), can help to understand the connection between the microscopic modeling of investors (agents) and the existence of stylized facts. One reason is the possibility to study the long time behavior of PDE models. This can be done by studying the steady state solutions of the PDE model. 
In the last decade, there have been several attemps from the physical and mathematical community to translate financial market models into time continuous PDE models. Examples are \cite{maldarella2012kinetic, cordier2009mesoscopic} and more recently \cite{Trimborn}. There are many mathematical methods to translate microscopic ordinary differential equations (ODEs) into  PDEs. A popular approach in economic applications is the kinetic Boltzmann method, mainly advanced by Toscani and Pareschi \cite{pareschi2013interacting}. This ansatz is mathematically well understood and has been applied to many different applications in life sciences and social sciences \cite{pareschi2013interacting}. In this paper, we follow a closely related approach. Instead of considering the kinetic Boltzmann description, we perform the mean field limit. The mean field limit is one of the classical kinetic limits as well, and describes the limit of infinitely many microscopic agents. The reason for that choice is that the microscopic coupling of agents in financial market models is induced by averaging of the agents' investment decisions. Hence, no binary interactions among agents, as considered in the kinetic Boltzmann approach, but rather a force field induced by the actions of all investors drives the microscopic dynamics. \klabsatz
The goal of this work is to derive the mean field limit of a microscopic econophysical financial market model. We show that the mesoscopic model is a good approximation of the original agent-based model.  
We have chosen a microscopic econophysical model which considers behavioral aspects of investors and reveals that they are the reasons for the existence of stylized facts. Thus, the starting point of our investigations is an agent-based model published by Cross et al. \citep{cross2005threshold} in 2005.
In Monte Carlo simulations, this model can reproduce the most prominent \textit{stylized facts} of financial data, namely: \textit{fat-tails in stock price return data, uncorrelated price returns and volatility clustering}. The benefit of this model is that it shows, by means of computer simulation, that the psychological herding pressure of investors causes the appearance of fat-tails . 
In absence of the herding pressure, the stock price behavior is characterized by a Gaussian return distribution. The financial agents are modeled as bounded rational agents in the sense of Simon \cite{simon1972theories}.
Each financial agent is described by his investment decision on the stock market, meaning if they are in a short position (sell stocks) or long position (buy stocks).
Since each agent is characterized through two possible orientations, there is an obvious connection to the Ising spin model \citep{ising1925beitrag} known in statistical physics. 
It is important to emphasize that this model follows a bottom up approach \citep{cross2005threshold} and develops the aggregated stock price behavior from reasonable microscopic interactions. These microscopic interactions can be interpreted as a simple trading strategy of investors. Nevertheless, we still consider a basic model. 
The model is simple in the sense that there are no binary interactions between agents and the agents have no learning ability \citep{cross2005threshold}. Furthermore, we want to underline that the goal of Cross et al. is not to predict prices or understand how to fit historic prices best \citep{cross2005threshold}, but rather gain insights into the creation of stylized facts. 
We also want to mention that there has been an earlier attempt to derive a mean field model of the original model \citep{cross2006mean}. In fact, the mean field model \citep{cross2006mean} follows a different philosophy and does not use a bottom up approach. In particular, the model \citep{cross2006mean} is substantially different compared to the original model  \citep{cross2005threshold}.
Before we can derive the mean field limit of the model of Cross et al., we need to ensure that there are no finite size effects regarding the number of agents.
Earlier studies \cite{egenter1999finite, zschischang2001some, kohl1997influence, hellthaler1996influence} and recently by Otte et al. \cite{SABCEMM} show that many agent-based econophysical models have finite size effects. Pleasingly, this is not the case in the original Cross model, as the simulations of Otte et al. \cite{SABCEMM} revealed. In this study, simulations with up to five million agents of the original Cross model have been conducted. For further simulation results we refer to the original papers \citep{cross2005threshold, cross2007stylized, lamba2008market} and the recently introduced SABCEMM tool \cite{SABCEMM}. \klabsatz
The result of the kinetic limit is a system of PDEs coupled with a stochastic differential equation (SDE). 
The SDE defines the time evolution of the market price, whereas the PDEs governs the investment decision of agents. We derive the space-homogeneous PDE-SDE system and the space-heterogeneous PDE-SDE system. The former corresponds to the \emph{rational} agent model with no herding, whereas the latter takes herding into account. In fact, the space-homogeneous model generates Gaussian stock price data while the heterogeneous model can create fat-tails in asset returns. Thus, our mesoscopic model exhibits the same characteristics and can reproduce  qualitatively the same stylized facts as the original model \cite{cross2005threshold}. Additionally, to its economic relevance this PDE-SDE system is already interesting for pure mathematical considerations. This model is very similar to models of animal aggregation originally introduced by Eftimie et al.  \cite{eftimie2007modeling}. Furthermore, our model is closely related to \emph{structured population dynamics} as discussed in \cite{michel2005general, perthame2005general}. This type of kinetic model has been probably first introduced by Kac \cite{kac1974stochastic}. In addition, there is also an obvious similarity to the famous Goldstein-Taylor model \cite{goldstein1951diffusion, taylor1922diffusion} .
\klabsatz
The structure of the paper is as follows: First, we introduce the econophysical model of Cross et al. \citep{cross2005threshold} which we denote \emph{Cross model}. 
We then introduce a microscopic approximation of the original Cross model which we call \emph{kinetic particle model}. In section 4, we derive the time continuous space-homogeneous and space-heterogeneous \emph{mean field Cross model}. Throughout the paper we give numerical simulations of the different models. In section 5, we extensively study the mean field Cross model numerically. In addition, we show that the mean field Cross model is qualitatively identical to the original Cross model. In section 6, we present a qualitative study of the mesoscopic model. We finish the paper with a short discussion of this work and a presentation of further research directions. \klabsatz

\section{The original Model}
In this section, we provide a brief definition of the original Cross model \citep{cross2005threshold}. 
There is a fixed number of $N\in\N$ agents. Each agent has to decide in each time step whether he wants to be long or short in the market, meaning if he wants so buy or sell stocks. Thus, the investment propensity $\gamma_i$ of each agent switches between a buy position $\gamma_i= 1$ and a sell position $\gamma_i=-1$. The excess demand function at time $t\in [0,\infty)$ is defined as the average of all investment decisions $\gamma_i$.
\begin{align}
ED_N(t):= \frac{1}{N}\sum\limits_{i=1}^N\gamma_i(t). \label{ED}
\end{align}
Hence, $ED_N$ measures the fraction of long respectively short investors. Furthermore, the model introduces two \textit{psychological} pressures, the \textit{herding pressure} and the \textit{inaction pressure}, which control the switching mechanism of investment decisions. The \textit{inaction pressure} is defined by the interval
\[
I_i=\left[ \frac{m_i}{1+\alpha_i}, m_i\ (1+\alpha_i)\right],
\] 
where $m_i$ denotes the stock price of the last switch of agent $i$ and $\alpha_i>0$ is the so called \textit{inaction threshold}. The investor switches position if the current stock price $S(t)>0$ leaves the interval $I_i$. This trading strategy can be interpreted as an agent \textit{taking his profits} or \textit{cutting his losses}. The model is discrete in time with fixed increments of time $\Delta t>0$.\\ 
The \textit{herding pressure} is given by:
\begin{align}
\begin{cases} c_i(t+\Delta t)= c_i(t)+ \Delta t\ |ED_N(t)|,& \text{if}\ \gamma_i(t)\ ED_N(t)<0,\\ \label{uh}
			c_i(t+\Delta t)=c_i(t),& \text{otherwise}.
			 \end{cases}
\end{align} 
Thus, the herding pressure is increased if the financial agent is in the minority position. 
The switch is induced if the herding pressure exceeds the \textit{herding threshold} $\beta_i$. The thresholds $\alpha_i, \beta_i$ are drawn from uniformly, independently and identically distributed random variables. 
\begin{align*}
\alpha_i\sim \Unif(A_1,A_2), \ A_2>A_1>0,\\
\beta_i\sim \Unif(B_1,B_2), \ B_2>B_1>0.
\end{align*}
We assume that $\alpha_i$ and  $\beta_i$ are uncorrelated and fixed after the initial choice. The constants
$B_1$ and $B_2$ have to scale with time, since they correspond to the time units an investor can resist the herding pressure.
\begin{align*}
& B_1:= b_1\cdot \Delta t,\ b_1>0,\\
&B_2:= b_2\cdot \Delta t,\ b_2>0.
\end{align*}
In summary, the switching mechanism can be described as follows.\\
 The switch is induced if 
	\[
	c_i>\beta_i\ \text{or}\  S(t)\notin I_i.
	\]
After each switch the \textit{herding pressure} gets reset to zero and the memory variable $m_i$ gets updated to the current stock price.\\ 
The stock price is then driven by the excess demand:
\begin{align}
&S(t+\Delta t)=S(t)\ \exp\left\{ \left( 1+\theta\ |ED_N(t)| \right)\ \left(\sqrt{\Delta t}\ \eta-\frac{\Delta t}{2}\right)+\kappa\ \Delta t\ \frac{\Delta ED_N(t)}{\Delta t} \right\},\\ \label{op}
& \eta\sim \mathcal{N}(0,1),\\
&\Delta ED_N(t):=\frac{1}{N}\sum\limits_{i=1}^N \gamma_i(t)-\frac{1}{N}\sum\limits_{i=1}^N \gamma_i(t-\Delta t),
\end{align}
where the constant $\kappa>0$ is known as market depth measuring the impact of a change in excess demand on the stock price. The \emph{heteroskedasticity parameter} $\theta\geq 0$ models the impact of the excess demand on the random external market information modeled by the Gaussian random variable. The reason for that choice is that \textit{``periods of extreme market volatility often coincide with periods of extreme"} \cite{cross2005threshold} excess demand. We refer to the original papers \citep{cross2005threshold, cross2007stylized, lamba2008market} for further modeling details.

\subsection{Microscopic Simulations}
In this section, we shortly present the outcome of simulations of the original Cross model. 
We investigate the most prominent stylized facts of financial data, namely fat-tails of asset returns, uncorrelated stock price returns and volatility clustering.
A fat-tailed distribution is characterized by an algebraic decay of the tails of the distribution. This can be well illustrated by a qq-plot, where the data is fitted against a Gaussian distribution. 
Uncorrelated price returns and volatility clustering can be deduced from the auto-correlation function of stock price returns. The former corresponds to an auto-correlation of zero and the latter to a slow decaying positive correlation of absolute log-returns for increasing time lags.\klabsatz  
If only the inaction pressure is active (blue graph in figure \ref{firstOriginal}), we obtain Gaussian behavior of the stock price since the excess demand is approximately zero. In fact, this trading rule can be regarded as in some sense \textit{rational}. If the herding pressure is added, the stock price behavior rapidly changes and we obtain non-Gaussian return distributions (green graph in figure \ref{firstOriginal}).
\begin{figure}[h!]
\begin{center}
\includegraphics[width=0.9\textwidth]{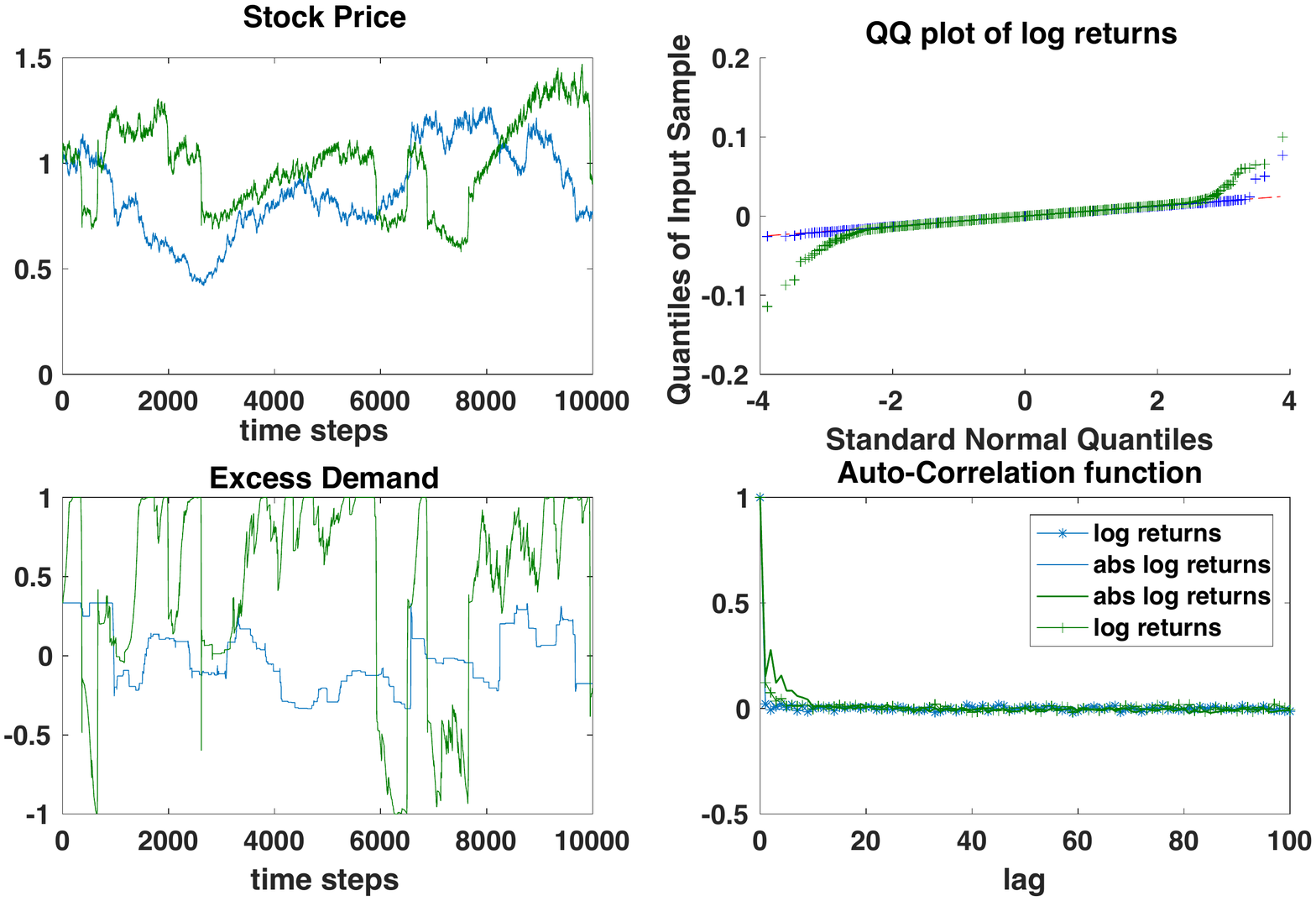}
\caption{Simulations of the original Cross model. The blue graph represents simulations only conducted with the inaction pressure, whereas the model output conducted with the inaction and herding pressure are colored green. The heteroskedasticity parameter is set to  $\theta=2$ and we refer to table \ref{ParCrossO} for further details.  }\label{firstOriginal}
\end{center}
\end{figure}
In both cases we observe uncorrelated raw price returns which can be deduced from the auto-correlation plot in figure \ref{firstOriginal}.
Furthermore, there is only a minor correlation in the case where both pressures are active. Figure \ref{secondOriginal} reveals, adding a dependence on the excess demand to the diffusion $(\theta=2)$, that we obtain volatility clustering. 
\begin{figure}[h!]
\begin{center}
\includegraphics[width=0.8\textwidth]{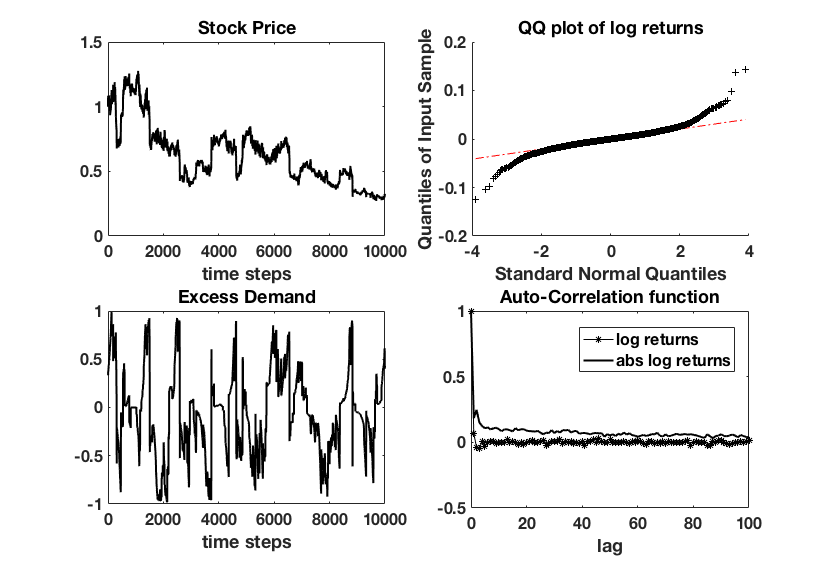}
\caption{Simulations of the original Cross model with inaction and herding pressure. The heteroskedasticity parameter is set to $\theta=2$. For further choices of parameters we refer to table \ref{ParCrossO}.}\label{secondOriginal}
\end{center}
\end{figure}
In agreement with earlier studies \cite{cross2005threshold, cross2007stylized}, we obtain no qualitative change in the simulations in cases where the herding thresholds are correlated or the thresholds get resampled after each switch.  For these reasons, we have used the simplest, previously introduced setting. \\
In summary, we want to record that the model with only active inaction pressure and constant diffusion function produces Gaussian stock return behavior. In comparison to that, adding the herding pressure changes the price behavior and we obtain fat-tails in the stock return. Furthermore, adding an excess demand depending diffusion function, results in volatility clustering.

\section{Kinetic Particle Model}
As pointed out previously, our goal is to derive a mesoscopic description of the agents' dynamics. From a mathematical perspective, the original Cross model is a highly non-linear dynamical system. In order to derive a kinetic PDE model, we need to consider the continuum limit $\Delta t\to 0$ and the mean field limit $N\to \infty$. The mean field limit, is well known in statistical physics and is concerned with the description of large particle systems by probabilistic quantities. Famous physical examples are the Ising model \cite{ising1925beitrag} or the Vlasov \cite{vlasov1938oscillation} equation.
Before we derive the kinetic model, we first derive a particle game. Secondly, we translate the particle game into a PDE. In fact, our financial agents are described by three quantities; the market position $\gamma_i\in\{-1,1\}$, the herding pressure $c_i\geq 0$ and the memory variable $m_i\geq 0$. The two quantities, the herding pressure and the memory variable, can take continuous values whereas the market position is discrete. Consequently, it is reasonable to divide the population into two groups, the agents holding a long position and the agents holding a short position. In a next step, we want to derive a switching probability of each agent to change his market position during a fixed time interval. This means that we want to neglect any dependencies of each agent on his personal past action. Hence, each agent rolls the dice at each time step and the switching probability of the agent only depends on the position of the agent in the $(m,c)$ space and the external stock price. 
This simplification is crucial in order to derive a kinetic PDE system. First, we aim to derive the switching probability based on the herding pressure denoted by $p(c)\in[0,1],\ c\in\R$. In fact, the herding thresholds are all realizations of a uniformly distributed random variable on $[B_1,B_2]$. Consequently, the probability for a switch is simply given by the cumulative distribution function of the random variable $\beta\sim \Unif(B_1,B_2)$ .  
\[
p(c):=P(\beta\leq c) = \int\limits_{\infty}^c \frac{x-B_1}{B_2-B_1}\ dx=\begin{cases}
  0,& c<B_1,\\
  \frac{c-B_1}{B_2-B_1}, &c\in [B_1,B_2],\\
  1,& c>B_2. 
  \end{cases}  
\]   
Equivalently, we define two random variables $$\psi \sim \Unif(M_1(m),M_2(m)),\ M_1(m):=\frac{m}{1+A_2},\ M_2(m):=\frac{m}{1+A_1}$$ and $$\eta\sim \Unif(M_3(m),M_4(m)),\ M_3(m):= m\ (1+A_1),\ M_4(m):=m\ (1+A_2)$$ for an arbitrary but fixed $m>0$. Notice that $M_1<M_2<M_3<M_4$ holds and consequently the probability of a switch for a given stock price $S>0$ and memory $m>0$
can be modeled by:
\begin{align*}
q(m,S):= 1-P(\psi\leq S)+P(\eta\leq S)=\begin{cases}
  1,& S<M_1(m),\\
  1-\frac{S-M_1(m)}{M_2(m)-M_1(m)}, &S\in [M_1(m),M_2(m)],\\
  0,& S\in (M_2(m),M_3(m)),\\
  \frac{S-M_3(m)}{M_4(m)-M_3(m)},& S\in [M_3(m),M_4(m)],\\
  1, & S> M_4(m).
  \end{cases}  
\end{align*} 
We define the switching probability to be the linear combination of these two probabilities. 
\[
\lambda_P(t,c,m,S):= \lambda_1\ p(c)+\lambda_2\ q(t,m,S),\ \lambda_1,\lambda_2>0,\ \text{with}\ \lambda_1+\lambda_2=1.
\]
This choice has been made partially for simplicity and as a result of simulations that indicate a good performance of this choice. We want to summarize the new kinetic particle model. \absatz
Each agent is described by the three properties $(\gamma_i,c_i,m_i)$. The excess demand is simply the average of the investment propensities \eqref{ED} and the time evolution of the herding pressure is given by \eqref{uh}.
At each time step $t_k:= k\ \Delta t,\ k\in\N$, the agent switches his market position with the probability $\lambda_P(t_k,c_i,m_i,S)$. The memory variable $m_i $ and the herding pressure $c_i$ get updated after a switch as in the original model. The time evolution of the stock price equation is given by
 \begin{align}
  S(t_{k+1})= S(t_k)+ \Delta t\ \kappa\  \frac{\Delta ED_N}{\Delta t}\ S(t_k)+ \sqrt{\Delta t}\ (1+\theta\ |ED_N|)\ S(t_k)\ \eta.\label{eulermaju}
 \end{align}
 The pricing equation \eqref{op} of the original Cross model approximates the underlying time continuous model by an explicit exponential integrator. 
 We approximate the time continuous SDE by an Euler-Maruyama discretization. \absatz 
In some sense, the derived switching probabilities can already be seen as the mean field limit of our system. In fact, our approximation is only good if we consider a sufficiently large number of agents. Thus this kinetic particle model should be seen as a realization of a random process of interacting agents. It is worthwhile to notice that this kinetic approximation leads to a noticeable reduction of computational costs in comparison to the original Cross model. The model outputs are qualitatively identical and both models compute the same quantities. The simulations have been conducted on the same machine. In our MATLAB implementation, we obtain a speedup of the factor $72$.

\subsection{Microscopic Simulations of Kinetic Particle Model}
In this section, we want to show that the kinetic particle model must be regarded a good approximation of 
the original Cross model at least on a qualitative level. Remember that we have introduced a switching rate for the investment decisions of agents and changed the pricing formula.\\
As in our previous microscopic simulations, we see a Gaussian behavior of the stock return distribution
in the pure inaction case (blue graph in figure \ref{KineticParticlefirst}).  When adding the herding pressure (green graph in figure \ref{KineticParticlefirst}) the behavior of the price rapidly changes and we obtain fat-tails in the price return distribution. 
\begin{figure}[h!]
\begin{center}
\includegraphics[width=0.9\textwidth]{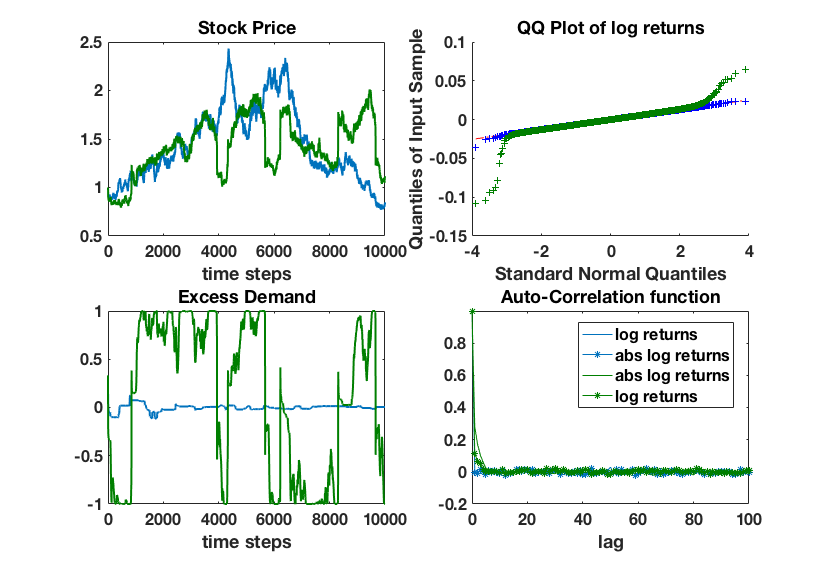}
\caption{Simulations of the kinetic particle model with inaction and herding pressure (green graph) and only inaction pressure (blue graph).
We have set the heteroskedasticity parameter $\theta =0$, for further details we refer to table \ref{ParKin}. }\label{KineticParticlefirst}
\end{center}
\end{figure}
These results coincide with the findings in the original Cross model. We obtain that the additional psychological herding effect forms jumps in the price process, respectively oscillations in the excess demand. \\
In the next simulation, see figure \ref{KineticParticlesecond}, we consider a positive heteroskedasticity parameter, setting $\theta=2$. As in the original Cross model, we obtain additional volatility clustering which can be deduced from the auto-correlation plot in figure \ref{KineticParticlesecond}.
\begin{figure}[h!]
\begin{center}
\includegraphics[width=0.9\textwidth]{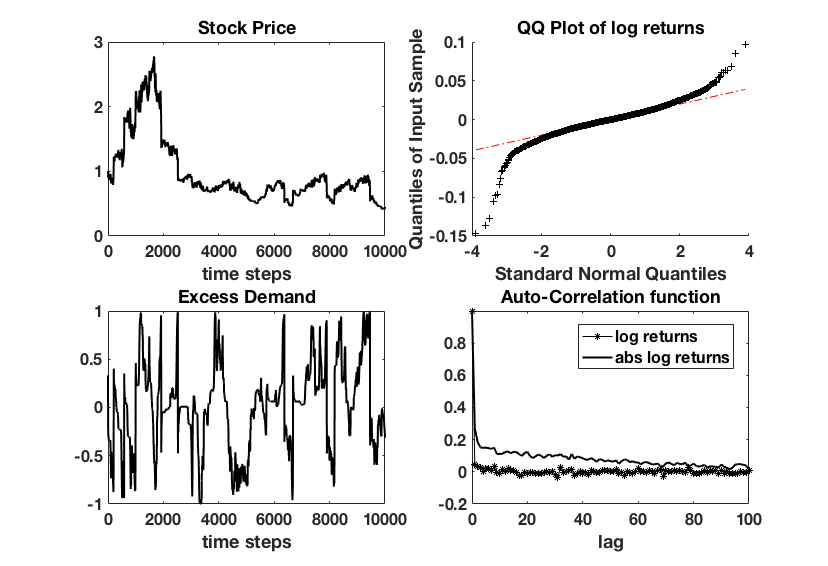}
\caption{Simulations of the kinetic particle model with inaction and herding pressure (green graph) and only inaction pressure (blue graph). We have set the heteroskedasticity parameter $\theta =2$, for further details we refer to table \ref{ParKin}.}\label{KineticParticlesecond}
\end{center}
\end{figure}

\section{Kinetic Model}
We introduce two groups of agents. One investor group is long $\gamma_i=1$ and the other short $\gamma_i=-1$ in the market. Hence, we consider the two densities $f^+(t,m,c)$ and $f^-(t,m,c)$.  
The variable $m\in\R$ we call memory variable which considers the stock price of the last switch. The variable $c\in\R$ is the herding pressure which is increased if the personal market position is in the opposite direction of the excess demand.
We have not chosen the half space $m,c\in \R_{\geq} 0$ to avoid non-standard boundary conditions at zero.  
Due to our choice, we can pose simple Dirichlet boundary conditions. Furthermore, we assume for a moment that the stock price $S(t)$ is externally provided. 
The time evolution of the densities is described by two phenomena:

\paragraph{Transport} There is an advection of the herding pressure which is proportional to the excess demand if the agent's decision contradicts the average opinion. Mathematically, this can be modeled by:
 \begin{align*}
&\partial_t f^+(t,m,c)+\partial_c\left(  H(-ED[f^+,f^-](t))\ f^+(t,m,c) \right)=0,\\
&\partial_t f^-(t,m,c)+\partial_c\left(H(ED[f^+,f^-](t))\  f^-(t,m,c) \right)= 0, 
\end{align*}
with 
\[
ED[f^+,f^-](t):={\int f^+(t,m,c)-f^-(t,m,c)\ dm\ dc}.
\]
Here, the shape function $H(\cdot)$ has the following properties
\begin{itemize}
\item $H(x)=0,\ \forall x\leq 0,$
\item $H(x)>0,\ \forall x>0, $
\item $\dot{H}(x)\geq 0,\ \forall x\geq 0$.
\end{itemize}
In order to approximate the original Cross model best, we choose the shape function as follows:
\[
H_C(x):=\begin{cases}
0,&\  x\leq0,\\
x,&,\ x>0.
\end{cases}
\]
The second effect of our particle dynamics are the interactions through the switching rate.
\paragraph{Switching Mechanism} 
As derived previously, the switching of investors between a long ($\gamma$=1)
and short $(\gamma=-1)$ position is fully determined by the switching rate $\lambda$. Since we are faced with a rate, we need to scale the probability $\lambda_P$ by the characteristic time step of the Cross model $\Delta t_C>0$, we define: $\lambda(t,c,m,S):= \frac{\lambda_P(t,c,m,S)}{\Delta  t_C}$. \\
The loss of agents with a long, respectively short position, is simply described by multiplication of the rate $\lambda$ with the corresponding density function. 
\begin{align*}
 Q_{loss}[f^{(\cdot)}](t,m,c,S):=f^{(\cdot)}(t,m,c)\ \lambda(t,m,c,S).
\end{align*}
The gain term is more complex. As determined by the particle dynamic, all agents which have switched are re-emitted in the $(c,m)$ space at the point $(0,S)$. Hence, translated into our continuous dynamics we get:  
\begin{align*}
 Q_{gain}[f^{(\cdot)}](t,m,c,S):=\delta(m-S(t))\ \delta(c)\ \int Q_{loss}[f^{(\cdot)}](t,m,c,S)\ dmdc.
 \end{align*}
Thus, e.g. in the case of $f^+$ the switching dynamic is given by:
\[
\partial_t f^+(t,m,c)=Q_{gain}[f^-](t,m,c,S)-Q_{loss}[f^+](t,m,c,S).
\]
\paragraph{The Model} Finally, the complete evolutionary dynamics of the densities $f^+,\ f^-$ are described by the system:  
 \begin{footnotesize}
\begin{align} 
\begin{split} \label{modelHetero}
&\partial_t f^+(t,m,c)+\partial_c\left(  H(-ED[f^+,f^-](t))\ f^+(t,m,c) \right)= Q_{gain}[f^-](t,m,c,S)-Q_{loss}[f^+](t,m,c,S),\\
&\partial_t f^-(t,m,c)+\partial_c\left(H(ED[f^+,f^-](t))\  f^-(t,m,c) \right)= Q_{gain}[f^+](t,m,c,S)-Q_{loss}[f^-](t,m,c,S).
\end{split}
\end{align}
\end{footnotesize}

This PDE system is coupled with the SDE
\begin{align}
dS = \kappa  \dot{ED}\ S\ dt+ (1+\theta\ |ED|)\ S\ dW, \label{SDE}
\end{align}
where $W$ denotes the Wiener process and we interpret the stochastic integral in the It\^o sense. 
The SDE \eqref{SDE} is the time continuous version of the previously introduced stock price equation \eqref{eulermaju}. 
The PDE-SDE systems is coupled through the excess demand $ED$. Besides initial conditions, we pose Dirichlet boundary conditions
\begin{align*}
& \lim\limits_{c\to\pm \infty} f^+(t,m,c)= \lim\limits_{c\to\pm \infty} f^-(t,m,c)=0,\\
& \lim\limits_{m\to\pm \infty} f^+(t,m,c)= \lim\limits_{m\to\pm \infty} f^-(t,m,c)=0.
\end{align*}
Thus, we can simplify the time derivative of the excess demand as follows:
\begin{footnotesize}
\begin{align*}
\frac{d}{dt} ED[f^+,f^-](t) &= \int \p{}{t} f^+(t,m,c)-\p{}{t}f^-(t,m,c)\ dmdc\\
&= \int - \partial_c\left(  H(-ED[f^+,f^-](t))\ f^+(t,m,c) \right)+\partial_c\left(H(ED[f^+,f^-](t))\  f^-(t,m,c) \right)\\
&\quad\quad +Q^+_{gain}(t,m,c,S)-Q^+_{loss}(t,m,c,S)- Q^-_{gain}(t,m,c,S)+Q^-_{loss}(t,m,c,S)\ dm dc\\
&=2\ \int f^-(t,m,c) \lambda(t,m,c,S)-f^+(t,m,c) \lambda(t,m,c,S)\ dmdc.
\end{align*}
\end{footnotesize}

\paragraph{Space-homogeneous Model}
In this paragraph, we define the space-homogeneous model. Here, we mean by space variable the herding variable $c$, although this is no physical space.
The reason for that choice is the analogy to kinetic theory, since the mean field Cross model has an advection in the herding variable $c$.
The investment decision does no longer depend on the two dimensional $(m,c)$ space but only on the memory variable $m$. Therefore, the space-homogeneous model does not include the herding effect and corresponds to the only inaction dynamics of the original Cross model. 
Thus, we define the space-homogeneous model by:    
 \begin{align}
 \begin{split}\label{homoModel}
&\partial_t g^+(t,m)=Q_{gain}^h[g^{-}] (t,m,S)-Q_{loss}^{h}[g^{+}](t,m,S),\\
&\partial_t g^-(t,m)=Q_{gain}^h[g^{+}] (t,m,S)-Q_{loss}^{h}[g^{-}](t,m,S), \\
& Q_{gain}^h[g^{(\cdot)}] (t,m,S):= \delta(m-S(t))\  \int g^{(\cdot)}(t,m)\ \lambda_h(t,m,S)\ dm,\\
&Q_{loss}^{h}[g^{(\cdot)}](t,m,S):= g^{(\cdot)}(t,m)\ \lambda_h(t,m,S).
\end{split}
\end{align}
The homogeneous model can be directly derived from the full model by integrating out the herding variable and setting the switching rate to $\lambda_h:=\frac{q}{\Delta t_C}$. Here, one uses the linearity of the collision integral. 
In fact, $g^{(\cdot)} (t,m):= \int\limits_{\R} f^{(\cdot)} (t,m,c)\ dc$ holds if $\lambda=\lambda_h$. 

\section{Numerics} \label{numerics}
In this section, we give numerical examples of our proposed mean field Cross model. We show that the kinetic model exhibits the same characteristic behavior as the original Cross model. We solve the PDE system with a standard finite volume discretization. We use a first order upwind scheme and apply the trapezoidal quadrature formula to evaluate the integrals of our model. The resulting ODEs are solved by an explicit Euler method. Notice that due to the stiff source term caused by the dirac deltas we get an additional stability condition to the classical Courant-Friedrich-Lewy condition.  We approximate the dirac deltas by a uniform distribution with support on one grid cell. The SDE is approximated by a simple Euler-Maruyama discretization. \\
First we present test cases of the space-homogeneous and secondly of the space-heterogeneous mean field Cross model.
Finally, we present the corresponding Monte Carlo solver of our mean field Cross model and give further examples. 
\paragraph{Space-homogeneous Model} 
Figure \ref{homo1} shows that in the space-homogeneous setting we obtain Gaussian distributed stock returns. 
Furthermore, there is no auto-correlation present (see figure \ref{homo1}).
\begin{figure}[h!]
\begin{center}
\includegraphics[width=1\textwidth]{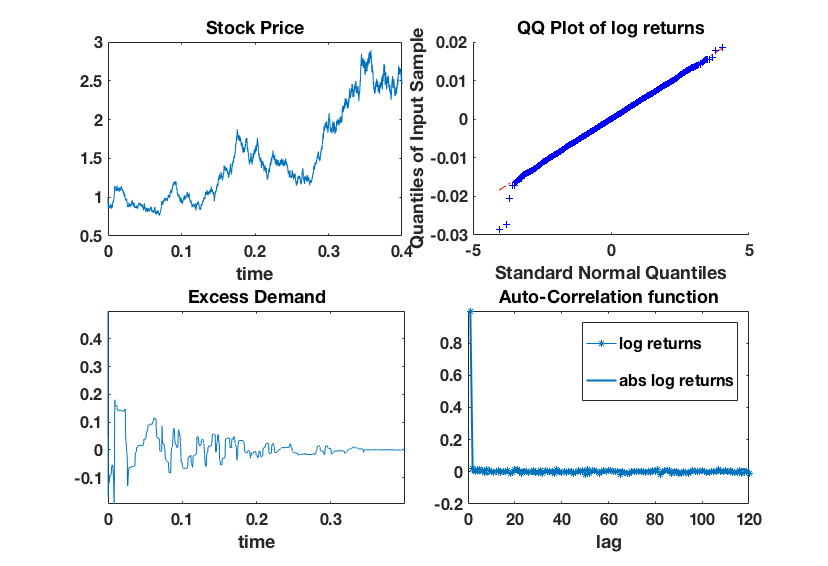}
\caption{Space-homogeneous model with $\theta=0$. For further parameters we refer to table \ref{ParMFCross}.}\label{homo1}
\end{center}
\end{figure}
The simulation in figure \ref{homo1} have been conducted with $\theta=0$ but we want to point out that for $\theta=2$ we obtain qualitatively the same result. 
We want to emphasize that the simulation results are qualitatively identical to the simulations of the original Cross model. 
\clearpage

\paragraph{Space-heterogeneous Model}
Figure \ref{hetero1} shows that the space-heterogeneous model output is characterized by a non-Gaussian return distribution. 
The heteroskedasticity parameter is set to zero and thus we see in figure \ref{hetero1} an auto-correlation of approximately zero.

\begin{figure}[h!]
\begin{center}
\includegraphics[width=1\textwidth]{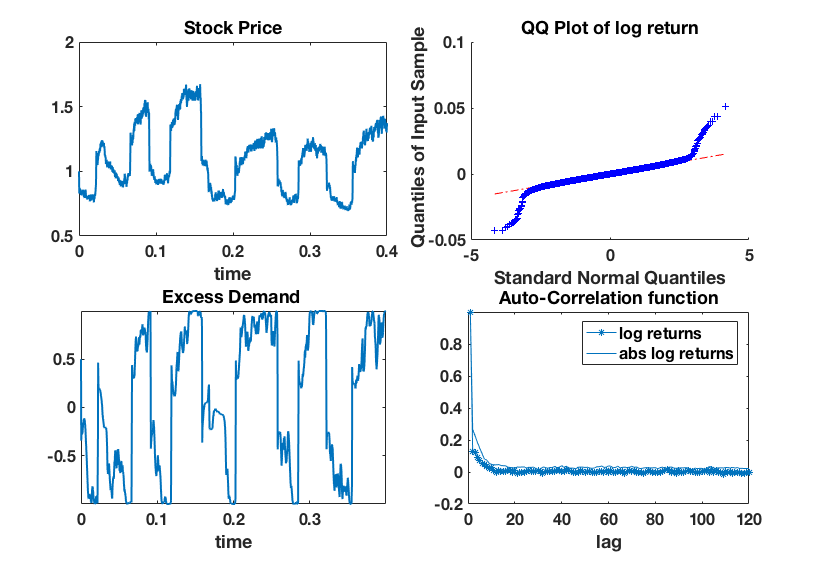}
\caption{Space-heterogeneous model with $\theta=0$. Further parameters are given in table \ref{ParMFCross}.}\label{hetero1}
\end{center}
\end{figure}
By setting the heteroskedasticity parameter to $\theta=2$ we then obtain a positive auto-correlation of absolute returns (see figure \ref{hetero2}).
The other characteristics of the model remain unchanged, thus figure \ref{hetero2} shows a non-Gaussian return distributions and oscillating excess demand as well. 
\begin{figure}[h!]
\begin{center}
\includegraphics[width=1\textwidth]{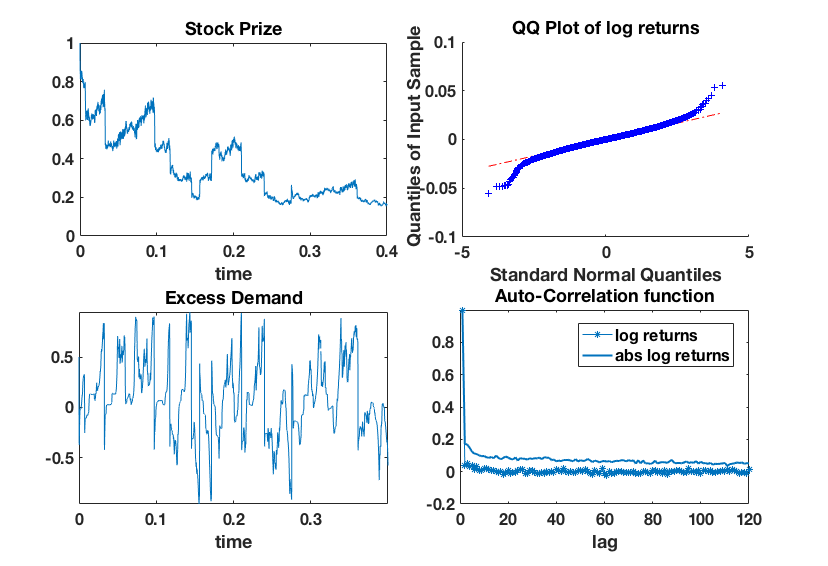}
\caption{Space-heterogeneous model with $\theta=2$. Further parameters are given in table \ref{ParMFCross}.}\label{hetero2}
\end{center}
\end{figure}
In summary, we can state that the mean field Cross model exhibits the same qualitative behavior as the original Cross model. 

\subsection{Deterministic Stock Price Equation}
The following test cases are conducted with a deterministic stock price equation. 
We do this in order to investigate if the homogeneous and heterogeneous models already behave differently 
in a fully deterministic setting. Thus, the mean field Cross model becomes a PDE-ODE system. 

\paragraph{Space-homogeneous Model}
From figure \ref{ODE1} we deduce that the dynamics converge to a steady profile. In addition, figure \ref{ODE1} reveals 
that the masses of both populations average before they reach a steady state. 
\begin{figure}[h!]
\begin{center}
\includegraphics[width=1\textwidth]{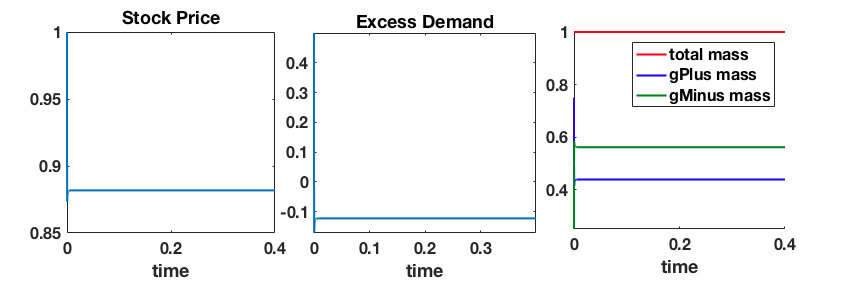}
\caption{Space-homogeneous model with deterministic stock price equation and $\theta=0$. For further parameters we refer to table \ref{ParMFCross}.}\label{ODE1}
\end{center}
\end{figure}

\paragraph{Space-heterogeneous Model}
As in the space-homogeneous model, the dynamics of the space-heterogeneous model reaches a steady state as well (see figure \ref{ODE2}).
In figure \ref{ODE2}, we obtain that the excess demand becomes $-1$. Thus, all agents have the same position. 
\begin{figure}
\begin{center}
\includegraphics[width=1\textwidth]{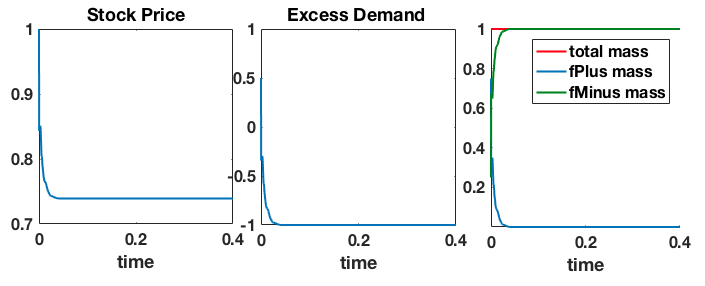}
\caption{Space-heterogeneous model with deterministic stock price equation and $\theta=2$. For further parameters we refer to table \ref{ParMFCross}.}\label{ODE2}
\end{center}
\end{figure}
Hence, we can conclude that already the deterministic skeleton of the space-homogeneous and space-heterogeneous model behave differently. In the next section, we will analyze the steady states of both models in detail.

\subsection{Monte Carlo Solver}
In this section we present a Monte Carlo solver of our space-heterogeneous mean field Cross model.  
Although the Monte Carlo solvers have a poor convergence rate, there is at least one advantage. In comparison to most
deterministic schemes, Monte Carlo solvers do not add any dissipation to the numerical solution \cite{pareschi2013interacting}.
This is an important feature, e.g. when analyzing the tail behavior of the density function. \\
In order to derive the Monte Carlo solver, we need to interpret the mean field Cross model as the master equation of a stochastic process. 
The main feature of the stochastic process is the switching mechanism. We summarize the Monte Carlo algorithm as follows. 
\begin{center}
\begin{tabular}{|l|}
\hline
\textbf{Monte Carlo Algorithm}\\
\hline
1. Generate sample $X_i^0=(\gamma_i^0,m_i^0,c_i^0)\in \{-1,1\} \times \R_{\geq 0}\times \R_{\geq 0},$\\
\quad $ \ i=1,...,N $ from initial distribution.\\
2. For each time step $k\in \{  1,...,\frac{T}{\Delta t}  \}$\\
\quad i) calculate $ED_N^k$ see \eqref{ED}, $\ S^{k+1}$ see equation \ref{eulermaju}, $\hat{\lambda}_i^k:=1-\exp(-\Delta t\ \lambda_i^k)$\\
\quad ii) update sample $X^k_i$ to $X_i^{k+1}$ and set\\
\quad\quad a) with probability $\hat{\lambda}_i^k$\\
\quad\quad\quad $\gamma_i^{k+1}=-\gamma_i^k,\ m_i^{k+1}=S^{k+1},\ c_i^{k+1}=0.$ \\
\quad\quad b) otherwise \\
\quad\quad\quad $\gamma_i^{k+1}=\gamma_i^k,\ m_i^{k+1}=m^{k}_i$\\
\quad\quad\quad and calculate $c_i^{k+1}$ by \eqref{uh}. \\
3. Reconstruct densities $f_{k+1}^+, f_{k+1}^-$. \\
\hline
\end{tabular} 
\end{center}
Notice that the switching rate of our particle model $\lambda_P=\Delta t_C\ \lambda$ is a first order Taylor approximation of $\hat{\lambda}$. The advantage of $\hat{\lambda}$ compared to 
$\lambda_P$ is that $\hat{\lambda}\in (0,1)$ holds for arbitrary time steps $\Delta t$. Thus, there are no time step restrictions which is an advantage compared to the finite volume method. 
\begin{figure}[h!]
\begin{center}
\includegraphics[width=1\textwidth]{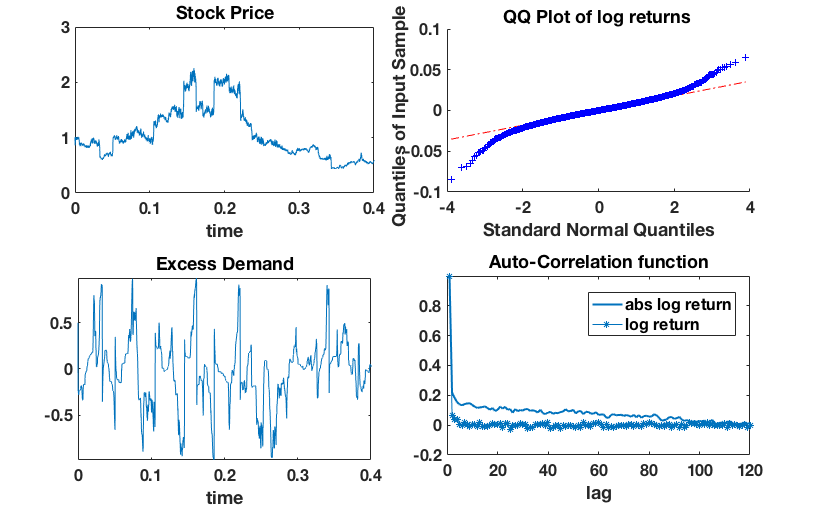}
\caption{Space-heterogeneous mean field Cross model with $\theta=2$. The simulation has been performed with a Monte Carlo solver with $10^5$ samples.
For further parameter settings we refer to table \ref{ParMFCross}.}\label{MC}
\end{center}
\end{figure}
Qualitatively, the model output in figure \ref{MC} coincide with the previous simulations conduced with a finite volume scheme.
Consequently, the results in figure \ref{MC} coincide qualitatively to the results of the original Cross model.


\section{Qualitative Behavior of the Model}
In this section, we want to study the analytical behavior of the space-homogeneous and space-heterogeneous mean field Cross model. The goal is to confirm our previous findings and to understand the complex model behavior in more detail.\\ \\
Both models are integro-differential equations equipped with a linear interaction integral. The number of agents is conserved for both models, which corresponds to the mass of the system. Furthermore, we prove that the only collision invariants of both models are given by constant functions. We refer to the appendix for details. \\ \\
 We divide our analysis in two parts, first we study the space-homogeneous model and secondly the space-heterogeneous model. 
\paragraph{Space-homogeneous Model}
The null space at time $t$ and stock price $S$ of the collision operator $Q^h[g](t,m,S):= Q_{gain}^h[g](t,m,S)-Q_{loss}^h[g](t,m,S)$ is given by
\[
\mathit{N}(Q^h)(t,S)=\{ g \in Y:  supp(g) \subseteq \{ m\in\R : \lambda_h(t,m,S) \equiv 0 \}     \},
\]
where $Y(\R,\R)$ denotes the set of young measures. We have chosen this function space since dirac delta functions are a subset of young measures. Notice that especially $g= \delta(m-S)$ is in the null space. \\
The previous simulations presented in figure \ref{ODE1} indicate that in the case of a deterministic stock price evolution our system \eqref{homoModel} reaches a steady profile. 
\paragraph{Steady States}
We assume that the stock price $S\equiv s_0>0$ is constant. This is reasonable because in equilibrium the excess demand is constant and thus the time derivative is zero. Hence, the right hand side of the deterministic stock price equation is zero.
Then all equilibrium solutions $g_{\infty}^+, g_{\infty}^-$ of the model are described by 

\begin{itemize}
\item[$a)$] $ED[g^+_{\infty}, g^-_{\infty}]=0$
\begin{itemize}
\item[$i)$] $g^+_{\infty}= g^-_{\infty}=0$.
\item[$ii)$] $g_{\infty}^+ , g_{\infty}^- >0$ and $g_{\infty}^+ , g_{\infty}^- \in \mathit{N}(Q^h)(s_0)$ with $\int g^+_{\infty}\ dm= \int g^-_{\infty}\ dm$.
\end{itemize}
\item[$b)$] $ED[g^+_{\infty}, g^-_{\infty}]<0$
\begin{itemize}
\item[$i)$] $g_{\infty}^+=0$ and $g_{\infty}^->0,\ g_{\infty}^- \in \mathit{N}(Q^h)(s_0)$.
\item[$ii)$] $g_{\infty}^+ , g_{\infty}^->0$ and $g_{\infty}^+ , g_{\infty}^- \in \mathit{N}(Q^h)(s_0)$ with $\int g^+_{\infty}\ dm \neq  \int g^-_{\infty}\ dm$.
\end{itemize}
\item[$c)$] $ED[g^+_{\infty}, g^-_{\infty}]>0$
\begin{itemize}
\item[$i)$] $g_{\infty}^-=0$ and $g_{\infty}^+>0,\ g_{\infty}^+ \in \mathit{N}(Q^h)(s_0)$.
\item[$ii)$] $g_{\infty}^+ , g_{\infty}>0$ and $g_{\infty}^+ , g_{\infty}^- \in \mathit{N}(Q^h)(s_0)$ with $\int g^+_{\infty}\ dm \neq  \int g^-_{\infty}\ dm$.
\end{itemize}
\end{itemize}
If the steady state solutions $g_{\infty}^+, g_{\infty}^-$ are elements of the null space $\mathit{N}(Q^h)(s_0)$, this means that they do not switch their market position any longer. The reason is that the memory variable or more precisely the stock price of the last switch is sufficiently close to the equilibrium price such that the agent does not feel the tension to change position.

\paragraph{Entropy Bound} As frequently done in kinetic theory, we want to show the entropy dissipation of our system \eqref{homoModel}. 
Such an entropy inequality is the key ingredient in order to prove uniqueness and asymptotic behavior in kinetic models. 
Mathematically, an entropy of a kinetic equation is a special kind of Lyapunov functional. 
We use the notion of general relative entropy \cite{michel2005general, perthame2005general}.
The dual equation of our system for a constant stock price $S\equiv s_0$ is given by
\begin{align}
\begin{split}\label{dualHomo}
&-\partial_t \psi^+(t,m) = \psi^-(t,S)\ \lambda^h(m,S) -\psi^+(t,m)\ \lambda^h(m,S) \\
&-\partial_t \psi^-(t,m) = \psi^+(t,S)\ \lambda^h(m,S) -\psi^-(t,m)\ \lambda^h(m,S) .
\end{split}
\end{align}
We define the general relative entropy for positive functions $\psi$, $p$ and a convex function $K$ to be
$$
t\mapsto \mathcal{K}_{\psi}(g,p):= \int \limits_{\R} \psi\ p\ K\left(\frac{g}{p}\right)\ dm.
$$
We can then formulate the following theorem. 
\begin{theorem}\label{entropy}
For all convex functions $K: \R\to \R$ and any solutions $p^+,p^->0,\ g^+,g^- $ of equation \eqref{homoModel} and solutions $\psi^+,\psi^->0$ of the dual equation \eqref{dualHomo}, the general relative entropy inequality
\begin{align}
\frac{d}{dt} \left(  \mathcal{K}_{\psi^+} (g^+,p^+)+ \mathcal{K}_{\psi^-} (g^-,p^-) \right)\leq 0,\label{entropyIn}
\end{align}
holds.
\end{theorem}
For the detailed proof we refer to the appendix \ref{AppendixAnal}. As a direct consequence, we can state the following a-priori bound:
\begin{lemma}
For any functions satisfying theorem \ref{entropy} and any convex function $K$ we can state the following inequality.
\begin{align*}
&\int\limits_{\R}  \psi^+(t,m)\ p^+(t,m)\ K\left( \frac{g^+(t,m)}{p^+(t,m)}\right)+   \psi^-(t,m)\ p^-(t,m)\ K\left( \frac{g^-(t,m)}{p^-(t,m)}\right)\ dm \\
\quad\quad &\leq \int\limits_{\R}  \psi^+(0,m)\ p^+(0,m)\ K\left( \frac{g^+(0,m)}{p^+(0,m)}\right)+   \psi^-(0,m)\ p^-(0,m)\ K\left( \frac{g^-(0,m)}{p^-(0,m)}\right)\ dm
\end{align*}
\end{lemma}

\begin{remark}
The application of the entropy inequality \eqref{entropyIn} in order to prove convergence to the equilibrium distributions is not straightforward. The difficulty is that all possible steady state solutions are not strictly positive. We expect, the long time asymptotics of system \eqref{homoModel} to be determined by the positive eigenvector of the largest non-negative eigenvalue \cite{perthame2005general}. The entropy inequality in theorem \ref{entropy} is the appropriate tool to show the long time convergence or even the rate of convergence to a steady state. 
\end{remark}

\paragraph{Space-heterogeneous Model}
We can analyze the heterogeneous model in the same manner as before. The null space at time $t$ and stock Price $S$ of the collision operator
$$
Q[f] (t,m,c,S):= Q_{gain}[f](t,m,c,S) - Q_{loss}[f](t,m,c,S),
$$ 
is given by 
\[
\mathit{N}(Q)(t,S)=\{ f \in \bar{Y}:  supp(f) \subseteq \{ (m,c)\in\R^2 : \lambda(t,m,c,S) \equiv 0 \}     \},
\]
where $\bar{Y}(\R^2,\R)$ is again the set of young measures. \\

\paragraph{Steady States} As before, we assume $S\equiv s_0>0$ and all equilibrium solutions $f^+_{\infty}, f^-_{\infty}$ are characterized by:

\begin{itemize}
\item[$A)$] $ED[f^+_{\infty}, f^-_{\infty}]=0$
\begin{itemize}
\item $f^+_{\infty}= f^-_{\infty}=0$.
\item $f_{\infty}^+ , f_{\infty}^- >0$ and $f_{\infty}^+ , f_{\infty}^- \in \mathit{N}(Q)(s_0)$ with $\int f^+_{\infty}\ dmdc= \int f^-_{\infty}\ dmdc$.
\end{itemize}
\item[$B)$] $ED[f^+_{\infty}, f^-_{\infty}]<0$
\begin{itemize}
\item $f_{\infty}^+=0$ and $f_{\infty}^->0,\ f_{\infty}^- \in \mathit{N}(Q)(s_0)$.
\end{itemize}
\item[$C)$] $ED[f^+_{\infty}, f^-_{\infty}]>0$
\begin{itemize}
\item $f_{\infty}^-=0$ and $f_{\infty}^+>0,\ f_{\infty}^+ \in \mathit{N}(Q)(s_0)$.
\end{itemize}
\end{itemize}
In comparison to the homogeneous setting, the case $ED[f_{\infty}^+, f_{\infty}^+]\neq 0$ with $f_{\infty}^+, f_{\infty}^->0$ cannot be a steady state. The reason is that due to the advection (increase of herding pressure), the partial derivative with respect to the herding pressure $c$ must be constant. Thus, for any test function $\phi(m,c)$
and $ED[f_{\infty}^+, f_{\infty}^+]<0$ 
$$
\int \phi(m,c)\ \partial_c\Big(H\big(-ED[f_{\infty}^+, f_{\infty}^+] \big)\ f_{\infty}^+\Big) \ dmdc=0,
$$
has to hold. This constant has to be zero, otherwise this would be a contradiction to our boundary condition. 
$$
\lim\limits_{m,c\to \infty} f_{\infty}^{(\cdot)} (m,c) =0.
$$
Hence, the excess demand can only take the values $\{-1,0,1\}$ in the equilibrium. This guides us to the explanation that the interplay of these steady states creates the characteristic oscillatory behavior in the stochastic simulations.

\paragraph{Entropy Bound} In the two dimensional case the definition of the generalized entropy can be translated one to one. Thus, for positive functions $\Phi$, $n$ and any convex function $K$ we have
$$
t\mapsto \mathcal{K}_{\Phi}(f,n):=\int \limits_{\R}  \int \limits_{\R} \Phi\ n\ K\left(\frac{f}{n}\right)\ dm dc.
$$
The dual equation of  the heterogeneous model is given by
\begin{align}
\begin{split}\label{dualHetero}
&-\partial_t \Phi^+(t,m,c)-H(-ED)\ \partial_c \Phi^+(t,m,c) = \Phi^-(t,S,c) \ \lambda(m, c,S) - \Phi^+(t,m, c)\ \lambda(m, c, S)\\
&-\partial_t \Phi^-(t,m,c) -H(ED)\ \partial_c \Phi^-(t,m,c)= \Phi^+(t,S,c)\ \lambda(m, c, S) - \Phi^-(t,m, c)\ \lambda(m,c ,S).
\end{split}
\end{align}

\begin{theorem}\label{entropyHetero}
For all convex functions $K: \R\to \R$ and any solutions $n^+,n^->0,\ f^+,f^- $ of equation \eqref{modelHetero} and solutions $\Phi^+,\Phi^->0$ of the dual equation \eqref{dualHetero} the general relative entropy inequality
\begin{align*}
\frac{d}{dt} \left(  \mathcal{K}_{\Phi^+} (f^+,n^+)+ \mathcal{K}_{\Phi^-} (f^-,n^-) \right)\leq 0,
\end{align*}
holds.
\end{theorem}
\begin{proof}
The proof is similar to the homogeneous case. The only difference is an additional advection term.
Due to the growth assumption on our densities the advection terms vanish after integration over $c$-space. 
\end{proof}
As direct consequence we get an a-priori bound. 
\begin{lemma}
For any functions satisfying theorem \ref{entropyHetero} and any convex function $K$ we can state the following inequality.
\begin{footnotesize}
\begin{align*}
&\int\limits_{\R}  \Phi^+(t,m,c)\ n^+(t,m,c)\ K\left( \frac{f^+(t,m,c)}{n^+(t,m,c)}\right)+   \Phi^-(t,m,c)\ n^-(t,m,c)\ K\left( \frac{f^-(t,m,c)}{n^-(t,m,c)}\right)\ dmdc \\
\quad\quad &\leq \int\limits_{\R}  \Phi^+(0,m,c)\ n^+(0,m,c)\ K\left( \frac{f^+(0,m,c)}{n^+(0,m,c)}\right)+   \Phi^-(0,m,c)\ n^-(0,m,c)\ K\left( \frac{f^-(0,m,c)}{n^-(0,m,c)}\right)\ dmdc.
\end{align*}
\end{footnotesize}
\end{lemma}

In order to study the stability properties of the space-heterogeneous system \eqref{modelHetero} and space-homogeneous system \eqref{homoModel}, we would need to analyze the eigenvalue problem. This study is left open for further research. As the steady state discussion of the space-heterogeneous and space-homogeneous model reveals, we expect to obtain fundamental different convergence and stability results in both models. 

\paragraph{Numerics}
This paragraph is devoted to confirm the findings and conjectures, we achieved in the previous investigations.\\
In the steady state discussion of the space-homogeneous model, we obtained that the steady state densities are identical zero or in the null space of the collision operator. 
Figure \ref{Supp} clearly shows that the dynamics are steady as long the support of the density functions is in the null space of the collision operator. 
\begin{figure}[h!]
\begin{center}
\includegraphics[width=.8\textwidth]{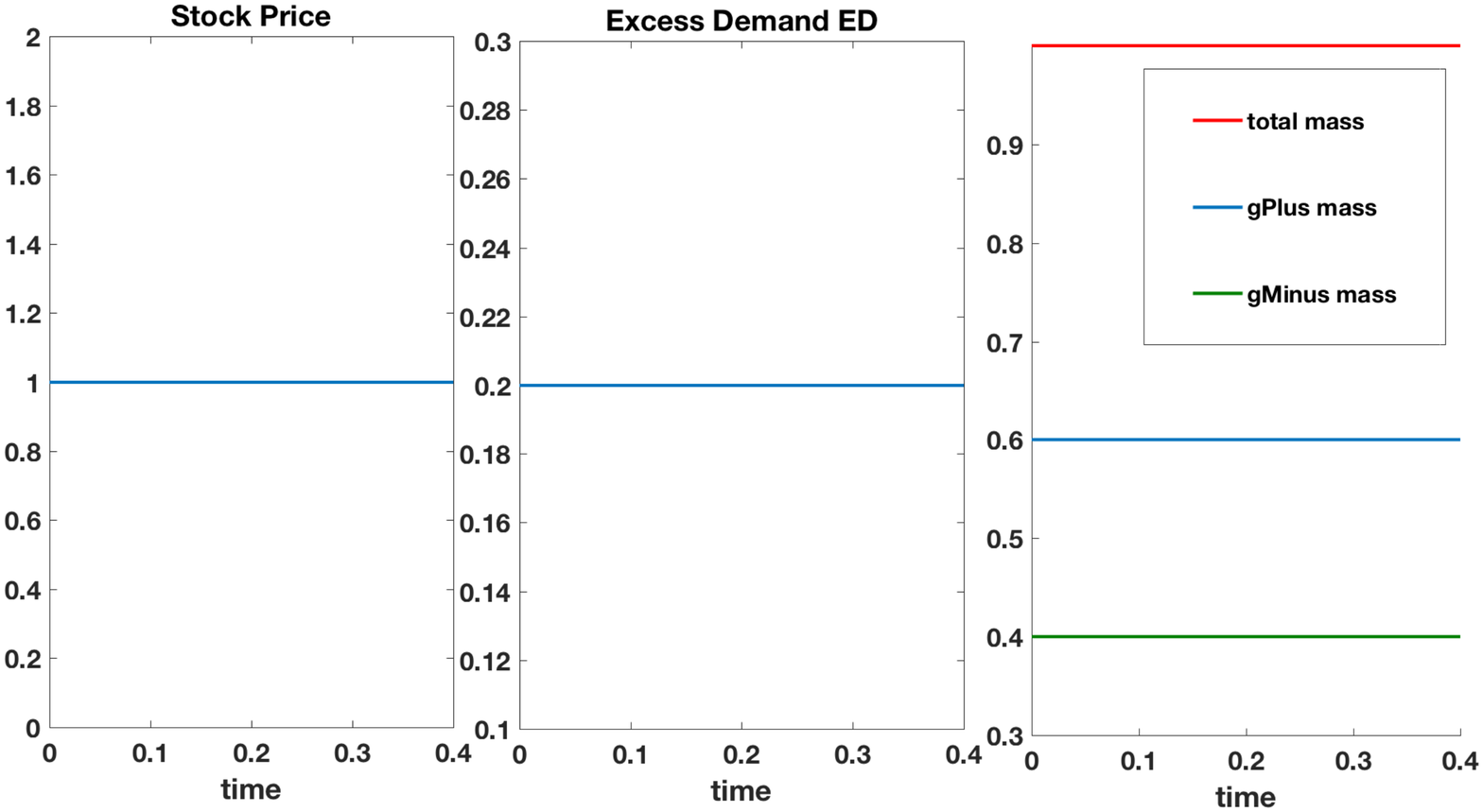}
\caption{Space-homogeneous deterministic mean field Cross model. The initial densities have their support in the null space of the collision operator. 
Further parameters are given in table \ref{ParMFCross}.}\label{Supp}
\end{center}
\end{figure}

\begin{figure}[h!]
\begin{center}
\includegraphics[width=1\textwidth]{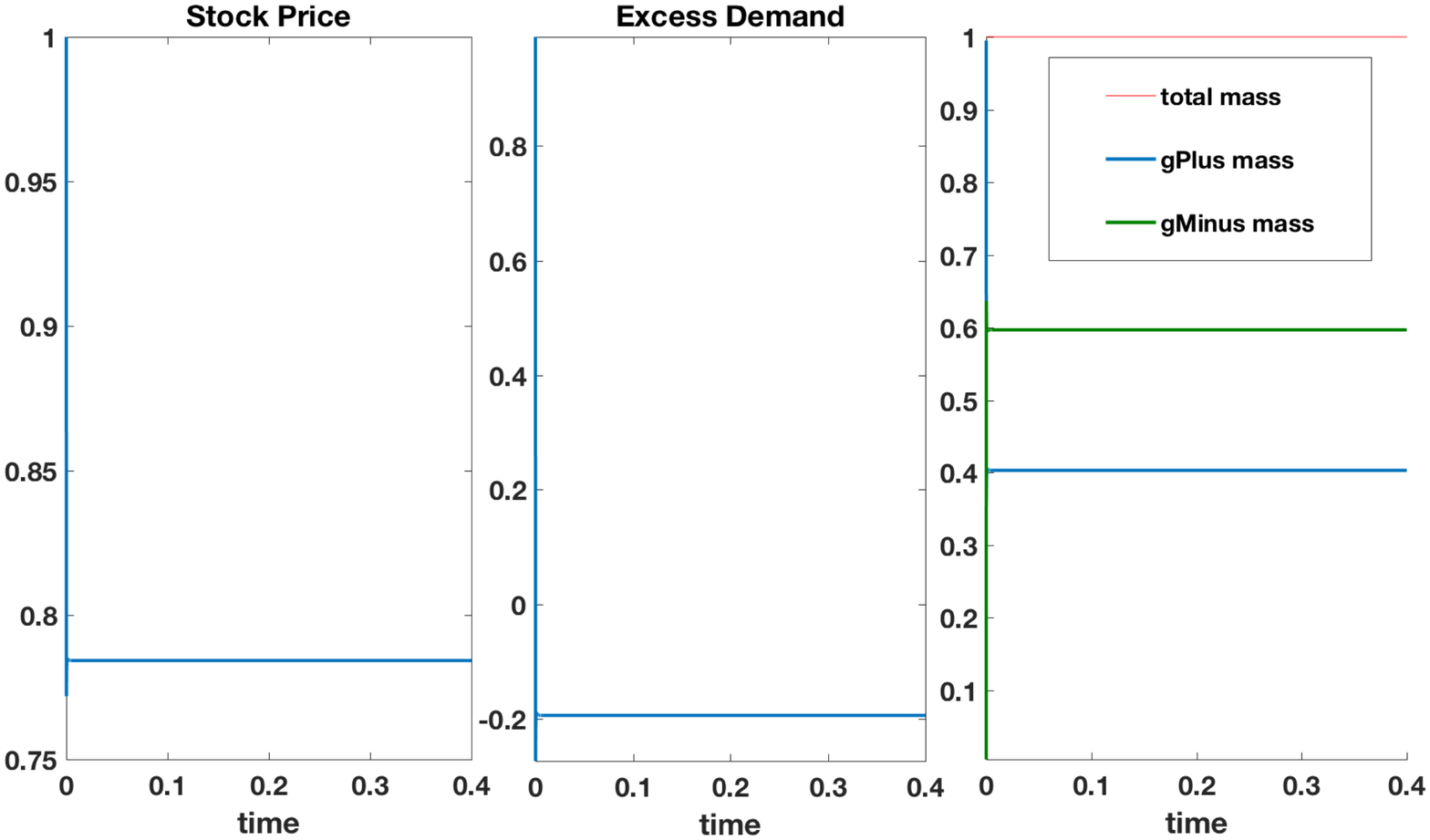}
\caption{Space-homogeneous deterministic mean field Cross model. The initial excess demand is given by $ED[g^+_0, g^-_0](0)=0.99$.
Further parameters are given in table \ref{ParMFCross}.}\label{Stability}
\end{center}
\end{figure}
The figures \ref{ODE1} and \ref{Stability} indicate the convergence of the solutions $g^+,g^-$ to the steady states $b)-ii)$ or $\ c)-ii)$ for general initial data. 
Interestingly, we see in figure \ref{Stability} a convergence to the steady state of type $b)-ii)$ although the initial mass of $g^-$ is close to zero. Thus, we conjecture that $b)-ii), c)-ii)$ are stable steady states, whereas $b)-i),\ c)-i)$ are unstable steady states. A proper proof of this numerical observation is left open for further research. \\ \\

The steady-state analysis of the space-heterogeneous model has shown that the excess demand can only reach the values $\{-1,0,1\}$ in equilibrium. 
Our deterministic simulations, visualized in figures \ref{ODE2}, \ref{StabH} and \ref{SuppH}, reveal that the excess demand always converges to the extreme values $\{1,-1\}$.
This is even the case if the initial densities have their support in the null space $\mathit{N}(Q)(t,S)$, see figure \ref{SuppH}.
This observation coincides with the steady states in $B)$ or $C)$ and again the initial values determine the convergence to one or another.
Furthermore, the results in figure \ref{StabH} indicate the stability of the steady states $B)$, respectively $C)$, in comparison to $A)$.
\begin{figure}[h!]
\begin{center}
\includegraphics[width=.8\textwidth]{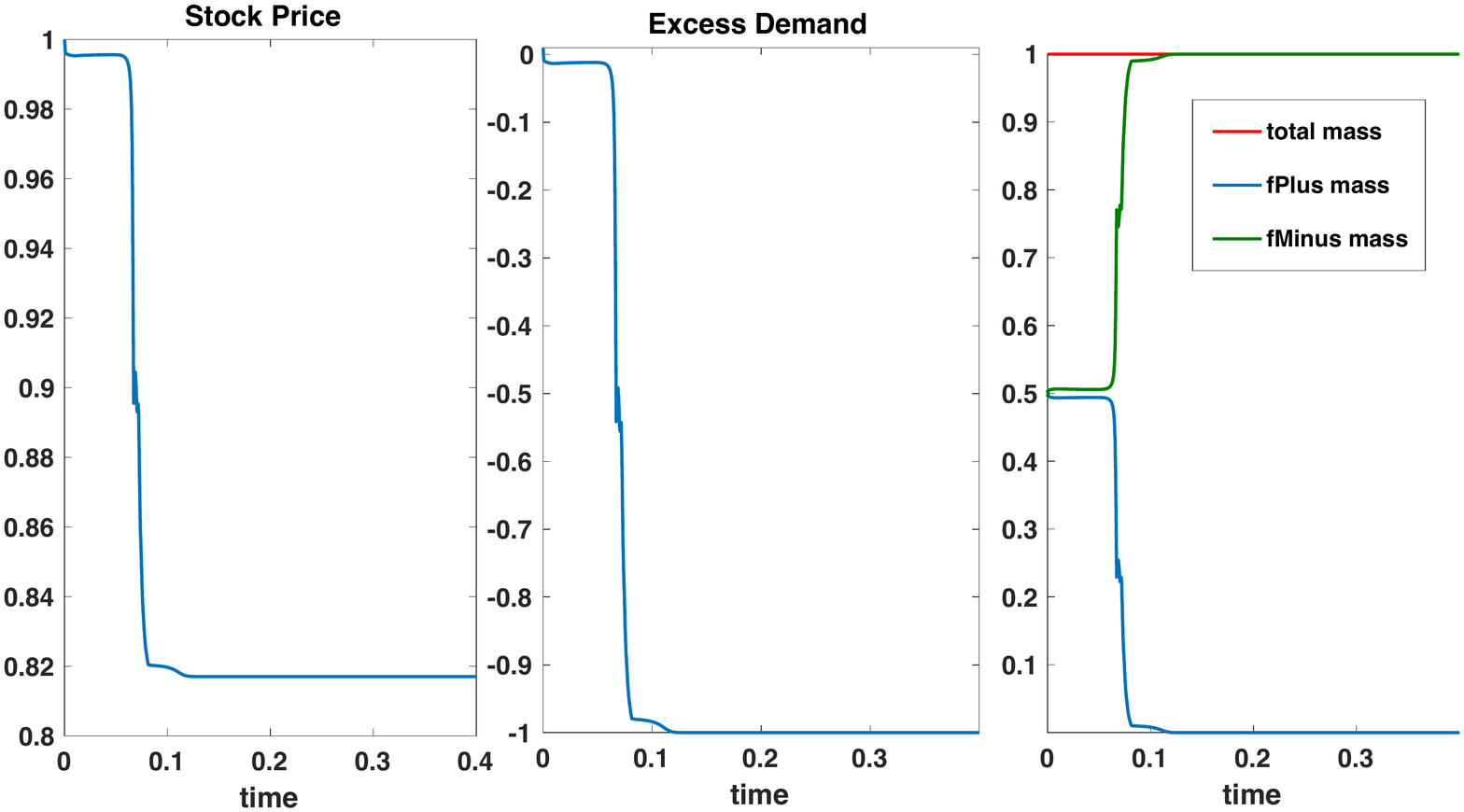}
\caption{Space-heterogeneous deterministic mean field Cross model. The initial excess demand is given by $ED[f^+_0,f^-_0](0) = 0.01$.
Further parameters are given in table \ref{ParMFCross}.}\label{StabH}
\end{center}
\end{figure}

\begin{figure}[h!]
\begin{center}
\includegraphics[width=1\textwidth]{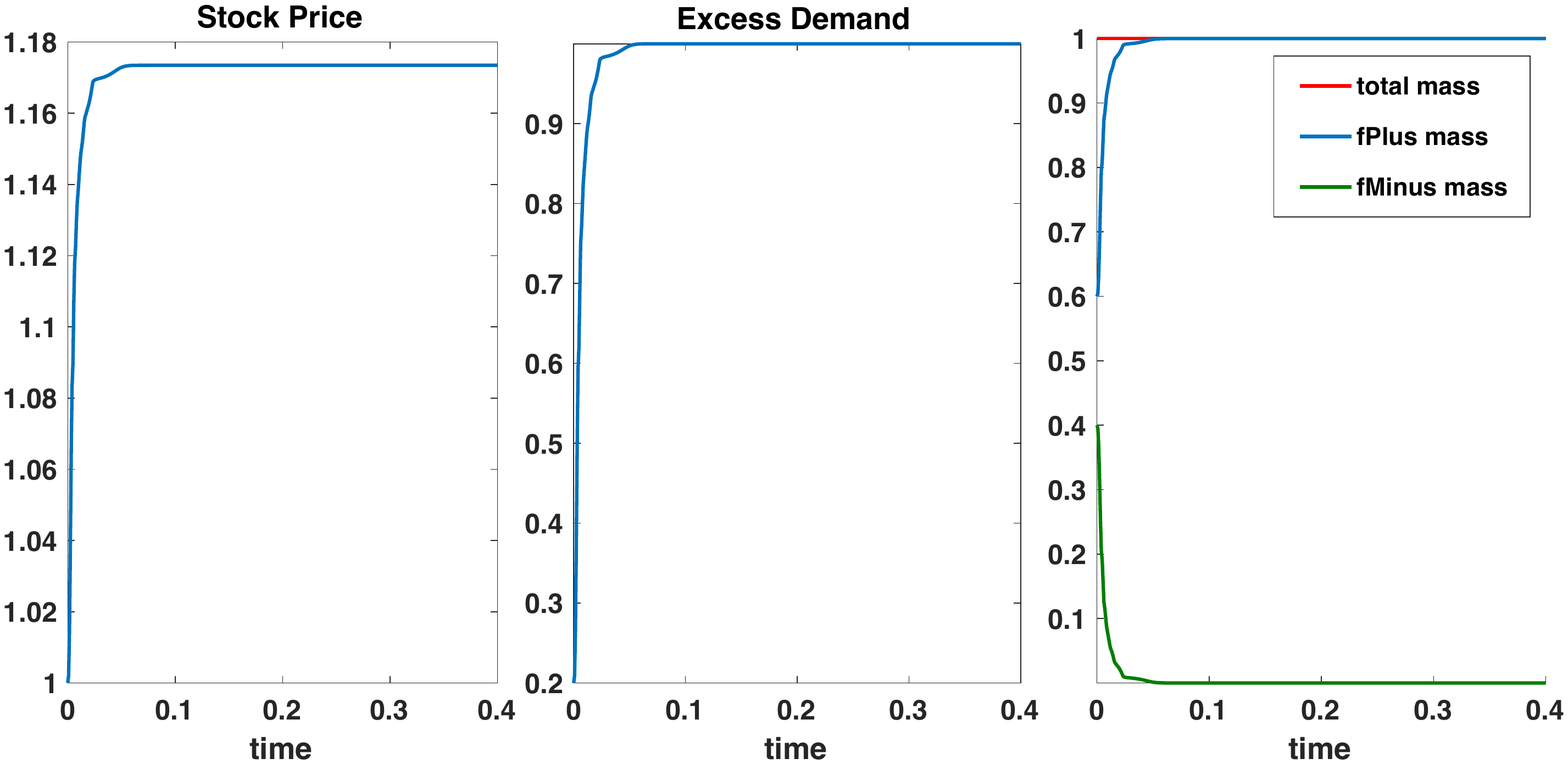}
\caption{Space-heterogeneous deterministic mean field Cross model. The initial densities have their support in the null space of the collision operator. 
For further parameter settings we refer to table \ref{ParMFCross}.}\label{SuppH}
\end{center}
\end{figure}

\section{Conclusion and Outlook}
We have introduced a kinetic model and have shown that the model is a good approximation of the original Cross model at least on a qualitative level. We have derived the continuum and mean field limit of the particle model and have obtained the mean field Cross model.
Our numerical investigations have revealed that the mean field Cross model exhibits identical characteristics as the original econophysical Cross model. The appearance of fat-tails is a direct consequence of the herding pressure. In the space-homogeneous case, where only the inaction pressure was active, we observe Gaussian behavior of stock returns. Interestingly, we only obtain volatitlity clustering in the space-heterogeneous model if we add the dependency of the diffusion function on the excess demand. \absatz
Furthermore, we have analyzed the PDE system with respect to steady states. We have shown that in the space-homogeneous case the excess demand can take various values compared to the heterogeneous case. In the heterogeneous case, the excess demand can only reach $\{-1,0,1\}$ in the equilibrium case. Hence, we conclude that this behavior of the deterministic skeleton explains the oscillatory behavior of the stochastic model. In addition, we could derive entropy bounds for the space-homogeneous and space-heterogeneous model as well. \absatz
We want to briefly discuss the advantages of the mean field Cross model compared to the microscopic Cross model. 
The PDE-SDE system most obviously enables us to do mathematical analysis, e.g. we could study the steady states of the deterministic PDE-ODE model.
Furthermore, we have gained a reduction of dimensions. Thus, instead of considering $N$ agents separately, we only consider two three-dimensional distribution functions. This can also be observed in the reduction of numerical complexity of the mean field model compared to the original Cross model. Of course we have to assume a large number of agents.
Finally, we also want to point out the reduction of parameters of the SDE-PDE system in contrast to the original microscopic model. \absatz
Further research directions are to quantify the influence of several parameters such as the market depth $\kappa$ on the statistical properties of the stock price. A sensitivity analysis or stochastic collocation might be performed to do uncertainty quantification. In addition, one might want to solve the inverse problem and fit several parameters to original stock price data. 
Again the advantages of the Cross model become obvious since we have reduced the number of unknowns remarkably.
Finally, we want to point out the possibility to extend the analysis of the model. Thus, the questions of existence, uniqueness 
and asymptotic convergence remain open. 

\section*{Acknowledgement}
The first-named author was supported by the Hans-Böckler-Stiftung. 

\clearpage
\appendix

\section{Appendix}

\subsection{Numerics}\label{Parameter}

\begin{table}[h!]
\begin{center}
\begin{tabular}{|c||c|}
\hline
Parameter & Value\\
\hline
\hline
$\kappa$ & $0.2$\\
\hline
$A_1$& $0.1$\\
\hline 
$A_2$ & $0.3$\\
\hline
$b_1$ & $25$\\
\hline
$b_2$ & $100$\\
\hline
$\Delta t$ &$ 4\cdot 10^{-5} $\\
\hline
$N$ & $1000$\\
\hline
Time Interval & $[0,0.4]$\\
\hline 
\end{tabular} 
\hspace{0.1cm}
\begin{tabular}{|c||c|}
\hline
 Variable & Initial Value\\
\hline
\hline
$\alpha$ & $0$ or $2$\\
\hline
$S(0)$ & $1$\\
\hline
$\gamma_i(0)$& $\gamma_i(0)=1,\ 1\leq i\leq 667,\ \gamma_i(0)=-1,\ 668\leq i\leq N$\\
\hline 
$ED(0)$& $\frac{1}{N}\sum\limits_{i=1}^N \gamma_i(0)$\\
\hline
$c_i(t)$ & $B_1,\ \forall 1\leq i\leq N$\\
\hline
$m_i(t)$ & $S(0),\ \forall 1\leq i\leq N$\\
\hline
\end{tabular}
\caption{Parameter settings of the original Cross model.}\label{ParCrossO}
\end{center}
\end{table}

\begin{table}[h!]
\begin{center}
\begin{tabular}{|c||c|}
\hline
Parameter & Value\\
\hline
\hline
$\kappa$ & $0.2$\\
\hline
$A_1$& $0.1$\\
\hline 
$A_2$ & $0.3$\\
\hline
$b_1$ & $25$\\
\hline
$b_2$ & $100$\\
\hline
$\Delta t$ &$ 4\cdot 10^{-5} $\\
\hline
$N$ & $30.000$\\
\hline
Time Interval & $[0,0.4]$\\
\hline 
\end{tabular} 
\hspace{0.1cm}
\begin{tabular}{|c||c|}
\hline
 Variable & Initial Value\\
\hline
\hline
$\alpha$ & $0$ or $2$\\
\hline
$\lambda_1,\lambda_2$ & $\lambda_1=\lambda_2=0.5$ or $\lambda_1=0,\ \lambda_2=1$\\
\hline
$S(0)$ & $1$\\
\hline
$\gamma_i(0)$& $\gamma_i(0)=1,\ 1\leq i\leq 667,\ \gamma_i(0)=-1,\ 668\leq i\leq N$\\
\hline 
$ED(0)$& $\frac{1}{N}\sum\limits_{i=1}^N \gamma_i(0)$\\
\hline
$c_i(t)$ & $B_1,\ \forall 1\leq i\leq N$\\
\hline
$m_i(t)$ & $S(0),\ \forall 1\leq i\leq N$\\
\hline
\end{tabular}
\caption{Parameter setting of the kinetic particle model.}\label{ParKin}
\end{center}
\end{table}

\begin{table}[h!]
\begin{center}
\begin{tabular}{|c||c|}
\hline
Parameter & Value\\
\hline
\hline
$\kappa$ & $0.2$\\
\hline
$A_1$& $0.1$\\
\hline 
$A_2$ & $0.3$\\
\hline
$b_1$ & $25$\\
\hline
$b_2$ & $100$\\
\hline
$\Delta t$ &$ 4\cdot 10^{-5} $\\
\hline
$N_c,N_m$ grid points & $400$\\
\hline
Time Interval & $[0,0.4]$\\
\hline 
\end{tabular} 
\hspace{0.5cm}
\begin{tabular}{|c||c|}
\hline
 Variable & Initial Value\\
\hline
\hline
$\alpha$ & $0$ or $2$\\
\hline
$\lambda_1,\lambda_2$ & $\lambda_1=\lambda_2=0.5$ \\
\hline
$S(0)$ & $1$\\
\hline
$ED(0)$& $\int f^+(0,m,c)-f^-(0,m,c)\ dmdc$\\
\hline
$f^+(0,m,c)$ & $ \Unif(M_1,m_4) \times \Unif(B_1,B_2)$\\
\hline
$f^-(0,m,c)$ &  $ \Unif(M_1,m_4) \times \Unif(B_1,B_2)$\\
\hline
\end{tabular}
\caption{Parameter settings of the mean field Cross model.}\label{ParMFCross}
\end{center}
\end{table}
\clearpage

\subsection{Qualitative Studies} \label{AppendixAnal}

We give the following definitions for space and velocity dependent distribution functions. They can be immediately transferred into the space-homogeneous setting.

\begin{definition}
Given any function $\phi(v,x),\ x,v\in\R^d$ and a density function $f(t,v,x),\ x,v\in\R^d$.
Then we call the average value with respect to the function $\phi(\cdot)$, an observable.
\begin{align*}
\langle  \phi(v,x), f(t,v,x)\rangle := \int\limits_{\R^d\times \R^d} \phi(v,x)\ f(t,v,x)\ dvdx.
\end{align*}
\end{definition}

\begin{definition}
We call the function $\psi(v,x)\in\R^n, \ x,v\in\R^d,\ n,d\in\N$ a collision invariant of the kinetic equation
\begin{align*}
\partial_t f(t,v,x)+ \nabla_x(G[f](t,v,x))=Q[f](t,v,x),
\end{align*} 
where $f: [0,\infty)\times \R^d \times \R^d\to \R^n$ is the density function,  $G[f](t,v,x)\in\R^n$ the flux and $Q[f](t,v,x)$ the collision operator, if
\begin{align*}
\int\limits_{\R^d}\psi(v,x)\cdot Q[f](t,v,x)\ dvdx =0, 
\end{align*}
holds for all functions $f$.
Furthermore, we call all observables of the kinetic density with respect to any collision invariant
\begin{align*}
\langle \psi(v,x), f(t,v,x)\rangle = \int\limits_{\R^d}\psi(v,x)\ Q[f](t,v,x) \ dvdx, 
\end{align*}
a \textbf{conserved quantity}. 
\end{definition}

\begin{remark}
Due to our growth assumption on the densities $f^+,f^-$ we get:
\begin{footnotesize}
\begin{align*}
&\partial_t \int\limits_{\R\times \R}  \phi_1(m,c)\ f^+(t,m,c)+ \phi_2(m,c)\ f^-(t,m,c)\  dmdc =\\
&\quad \quad  \int\limits_{\R\times \R}  \phi_1(m,c)\ \left( Q_{gain}[f^-](t,m,c,S)-Q_{loss}[f^+](t,m,c,S)\right)\ dmdc \\
 &\quad \quad +\int\limits_{\R\times \R} \phi_2(m,c)\ \left( Q_{gain}[f^+](t,m,c,S)-Q_{loss}[f^-](t,m,c,S)\right)\ dmdc.
\end{align*}
If $\phi_1,\phi_2$ are collision invariants we have:
\begin{align*}
&\partial_t \int\limits_{\R\times \R}  \phi_1(m,c)\ f^+(t,m,c)+ \phi_2(m,c)\ f^-(t,m,c) \ dmdc =0.
\end{align*}
Thus, the conserved quantities are constant in time, which reveals the motivation of their name.
\end{footnotesize}
\end{remark}

\begin{theorem}
All collision invariants of our homogeneous model are given by
\begin{align*}
\psi_1(m)=c_1,\\
\psi_2(m)=c_2,
\end{align*}
and in the space-heterogeneous case we get:
\begin{align*}
\psi_1(m,c)=c_1,\\
\psi_2(m,c)=c_2,
\end{align*}
where $c_1,c_2, \in\R$ are constants. Thus, the only conserved quantity of our system is the mass, respectively the number of agents. 
\end{theorem}

\begin{proof}
We perform the proof for the space-heterogeneous setting, but one can translate the results one to one to the one dimensional case. 
The functions $\psi_1(m,c), \psi_2(m,c)$ have to satisfy:
\begin{align*}
&\int\limits_{\R\times \R}  \psi_1(m,c)\ \left( Q_{gain}[f^-](t,m,c,S)-Q_{loss}[f^+](t,m,c,S)\right)\ dmdc +\\
&\quad\quad +\int\limits_{\R\times \R}  \psi_2(m,c)\ \left( Q_{gain}[f^+](t,m,c,S)-Q_{loss}[f^-](t,m,c,S)\right)\ dmdc=0
\end{align*}
This is equivalent to
\begin{align*}
&\psi_1(S,0)\ \int\limits_{\R\times \R} \lambda(t,c,m,S)\ f^-(t,m,c)\ dmdc-\int\limits_{\R^2} \psi_1(m,c)\ \lambda(t,c,m,S)\ f^+(t,m,c)\ dmdc+ \\
& +\psi_2(S,0)\ \int\limits_{\R\times \R} \lambda(t,c,m,S)\ f^+(t,m,c)\ dmdc-\int\limits_{\R\times\R} \psi_2(m,c)\ \lambda(t,c,m,S)\ f^-(t,m,c)\ dmdc=0
\end{align*}
This equation can be rewritten.
\begin{align*}
&\int\limits_{\R\times \R} \lambda(t,c,m,S)\ f^-(t,m,c)\ \left( \psi_1(S,0)-\psi_2(m,c)\right)\ dm dc\\
&+ \int\limits_{\R\times \R} \lambda(t,c,m,S)\ f^+(t,m,c)\ \left(\psi_2(S,0)-\psi_1(m,c)\right)\ dm dc=0.
\end{align*}
The previous equation has to hold for all functions $f$.
Thus, by the lemma of variational calculus, we can conclude that 
\begin{align*}
\psi_1(S,0)=\psi_2(m,c),\\
\psi_2(S,0)=\psi_1(m,c),
\end{align*} 
has to hold. Hence, we define $c_2:=\psi_1(S,0)$ and $c_1:=\psi_2(S,0)$ and the proof is completed. 
\end{proof}

The proof of theorem \ref{entropy} is given by:
\begin{proof}
A straightforward computations shows.
\begin{footnotesize}
\begin{align*}
&\partial_t \left[ \psi^+(t,m)\ p^+(t,m)\ K\left(\frac{g^+(t,m)}{p^+(t,m)}\right)\right] +\partial_t \left[ \psi^-(t,m)\ p^-(t,m)\ K\left(\frac{g^-(t,m)}{p^-(t,m)}\right)\right]\\
& \quad +\psi^+(t,S) p^-(t,m)\ K\left(\frac{g^-(t,m)}{p^-(t,m)}\right)\ \lambda^h(m,S)\\
& \quad -\psi^+(t,m)\ \delta(m-S)\ \int\limits_{\R} \lambda^h(t,m^{\prime})\ p^-(t,m^{\prime})\ K\left(\frac{g^-(t,m^{\prime})}{p^-(t,m^{\prime})}\right)\ dm^{\prime}\\
& \quad +\psi^-(t,S) p^+(t,m)\ K\left(\frac{g^+(t,m)}{p^-(t,m)}\right)\ \lambda^h(m,S)\\
&\quad -\psi^-(t,m)\ \delta(m-S)\ \int\limits_{\R} \lambda^h(t,m^{\prime})\ p^+(t,m^{\prime})\ K\left(\frac{g^+(t,m^{\prime})}{p^+(t,m^{\prime})}\right)\ dm^{\prime}\\
&= \psi^+(t,m)\ \delta(m-S)\ \int\limits_{\R} \lambda^h(m^{\prime},S)\ p^-(t,m^{\prime}) \Bigg( \left[ K\left( \frac{g^+(t,m)}{p^+(t,m)}\right) -  K\left( \frac{g^-(t,m^{\prime})}{p^-(t,m^{\prime})}\right)  \right]\\
& \quad\quad\quad\quad\quad\quad\quad\quad\quad + K^{\prime}\left( \frac{g^+(t,m)}{p^+(t,m)}\right) \left[\frac{g^-(t,m^{\prime})}{p^-(t,m^{\prime})}  - \frac{g^+(t,m)}{p^+(t,m)}\right]   \Bigg)\ dm^{\prime}\\
&+ \psi^-(t,m)\ \delta(m-S)\ \int\limits_{\R} \lambda^h(m^{\prime},S)\ p^+(t,m^{\prime}) \Bigg( \left[ K\left( \frac{g^-(t,m)}{p^-(t,m)}\right) -  K\left( \frac{g^+(t,m^{\prime})}{p^+(t,m^{\prime})}\right)  \right]\\
& \quad\quad\quad\quad\quad\quad\quad\quad\quad + K^{\prime}\left( \frac{g^-(t,m)}{p^-(t,m)}\right) \left[\frac{g^+(t,m^{\prime})}{p^+(t,m^{\prime})}  - \frac{g^-(t,m)}{p^-(t,m)}\right]   \Bigg)\ dm^{\prime}.
\end{align*}
Then we integrate over $m$ and get.
\begin{align*}
&\frac{d}{dt} \int\limits_{\R}  \psi^+(t,m)\ p^+(t,m)\ K\left(\frac{g^+(t,m)}{p^+(t,m)}\right)\ dm+\frac{d}{dt} \int\limits_{\R}  \psi^-(t,m)\ p^-(t,m)\ K\left(\frac{g^-(t,m)}{p^-(t,m)}\right)\ dm\\
&= \psi^+(t,S)\  \int\limits_{\R} \lambda^h(m^{\prime},S)\ p^-(t,m^{\prime}) \Bigg( \left[ K\left( \frac{g^+(t,S)}{p^+(t,S)}\right) -  K\left( \frac{g^-(t,m^{\prime})}{p^-(t,m^{\prime})}\right)  \right]\\
& \quad\quad\quad\quad\quad\quad\quad\quad\quad + K^{\prime}\left( \frac{g^+(t,S)}{p^+(t,S)}\right) \left[\frac{g^-(t,m^{\prime})}{p^-(t,m^{\prime})}  - \frac{g^+(t,S)}{p^+(t,S)}\right]   \Bigg)\ dm^{\prime}\\
&+ \psi^-(t,S)\ \int\limits_{\R} \lambda^h(m^{\prime},S)\ p^+(t,m^{\prime}) \Bigg( \left[ K\left( \frac{g^-(t,S)}{p^-(t,S)}\right) -  K\left( \frac{g^+(t,m^{\prime})}{p^+(t,m^{\prime})}\right)  \right]\\
& \quad\quad\quad\quad\quad\quad\quad\quad\quad + K^{\prime}\left( \frac{g^-(t,S)}{p^-(t,S)}\right) \left[\frac{g^+(t,m^{\prime})}{p^+(t,m^{\prime})}  - \frac{g^-(t,S)}{p^-(t,S)}\right]   \Bigg)\ dm^{\prime}.
\end{align*}
\end{footnotesize}
Thanks to the convexity of $K$ the inequality 
$$
0\geq K(y)-K(x)+K^{\prime}(y)\ (x-y),
$$
holds for any differentiable function $K$. Thus, the right hand side is negative and the entropy inequality \eqref{entropyIn} holds.
\end{proof}

\clearpage

	\bibliographystyle{abbrv}	
	\bibliography{literaturmean.bib}
\end{document}